%% file: lti_submit.tex
\documentclass[12pt]{article}
\usepackage{amssymb,theorem,amsmath}
\usepackage{epsf}
\usepackage[dvips]{graphicx}
\usepackage{pictex}

\newcommand{\lcm}{\mathop{\rm lcm}}

\def\csect#1{Section~\ref{#1}}
\def\csects#1#2{Sections~\ref{#1} and \ref{#2}}

\newcommand{\be}{\begin{equation}}
\newcommand{\ee}{\end{equation}}

\newtheorem{theorem}{Theorem}
\newtheorem{corollary}[theorem]{Corollary}
\newtheorem{lemma}[theorem]{Lemma}

{\theorembodyfont{\rm} \newtheorem{example}[theorem]{Example}
\newtheorem{remark}[theorem]{Remark}

 }

\newcommand{\cthm}[1]{Theorem~\ref{#1}}

\newcommand{\clem}[1]{Lemma~\ref{#1}}
\newcommand{\ccor}[1]{Corollary~\ref{#1}}
\newcommand{\cfig}[1]{Figure~\ref{#1}}
\newcommand{\crem}[1]{Remark~\ref{#1}}

\newcommand{\cex}[1]{Example~\ref{#1}}
\newenvironment{proof}{\noindent\textbf{Proof:} }
   {\hfill {\quad$\blacksquare$}\par\medskip}
\newenvironment{proofof}[1]{\medskip\noindent
   \textit{Proof of #1:} }{\hfill {\quad$\blacksquare$}\par\medskip}

    \def\ZZ{\mathbb{Z}}
    
    \def\bbz{\mathbb{Z}}
    \def\bbl{\mathbb{L}}
    \def\bbt{\mathbb{T}}

  \def\C{{\cal C}}\def\P{{\cal P}}

\def\nn{\nonumber}
\def\interiorM{\rlap{\raise9pt\hbox to9.5pt{\hss$\scriptscriptstyle\circ$}}M}

\begin{document}

\title{Translation invariant extensions of finite volume measures}
\author{S. Goldstein\footnote{Department of Mathematics,
Rutgers University, New Brunswick, NJ 08903.},
 T. Kuna\footnote{Department of Mathematics and Statistics,
      University of Reading, Whiteknights, PO~Box 220, Reading RG6 6AX, UK.},
J. L. Lebowitz\footnotemark[1],
\footnote{Also Department of Physics, Rutgers.}\ \ 
and E. R. Speer\footnotemark[1]}

\maketitle

\noindent{\raggedright
{\bf Keywords:} local translation invariance, translation invariant
extensions, maximal entropy extensions, de Bruijn graphs, Bethe lattice\par}

\medskip

\begin{abstract}We investigate the following questions: Given a measure
$\mu_\Lambda$ on configurations on a subset $\Lambda$ of a lattice $\bbl$,
where a configuration is an element of $\Omega^\Lambda$ for some fixed set
$\Omega$, does there exist a measure $\mu$ on configurations on all of
$\bbl$, invariant under some specified symmetry group of $\bbl$, such that
$\mu_\Lambda$ is its marginal on configurations on $\Lambda$?  When the
answer is yes, what are the properties, e.g., the entropies, of such
measures?  Our primary focus is the case in which $\bbl=\bbz^d$ and the
symmetries are the translations. For the case in which $\Lambda$ is an
interval in $\ZZ$ we give a simple necessary and sufficient condition, {\it
local translation invariance (LTI)}, for extendibility.  For LTI measures
we construct extensions having maximal entropy, which we show are Gibbs
measures; this construction extends to the case in which $\bbl$ is the
Bethe lattice.  On $\bbz$ we also consider extensions supported on periodic
configurations, which are analyzed using de~Bruijn graphs and which include
the extensions with minimal entropy.  When $\Lambda\subset\bbz$ is not an
interval, or when $\Lambda\subset\bbz^d$ with $d>1$, the LTI condition is
necessary but not sufficient for extendibility.  For $\bbz^d$ with $d>1$,
extendibility is in some sense undecidable.  \end{abstract}

\section{Introduction\label{intro}}

We consider an extension problem for measures describing configurations
on a lattice $\bbl$, in which the value $\eta_i$ of the configuration
$\eta$ at site (or equivalently vertex) $i$ of $\bbl$ belongs to some
fixed set $\Omega$.  Let $X=\Omega^{\bbl}$ be the set of all
configurations on $\bbl$ and, for $\Lambda\subset\bbl$, let
$X_\Lambda=\Omega^{\Lambda}$ be the set of configurations on $\Lambda$.
Then we address the following extension problem: for a fixed
$\Lambda\subset\bbl$ and probability measure $\mu_\Lambda$ on
$X_\Lambda$, does there exist a probability measure $\mu$ on $X$ such
that (i)~$\mu$ is invariant under some specified group of symmetries of
$\bbl$ and (ii)~the marginal measure of $\mu$ on $X_\Lambda$ is
$\mu_\Lambda$?  When $\mu$ exists we speak of it as an  {\it invariant
extension} of $\mu_\Lambda$ or, since we consider only invariant
extensions, simply as an {\it extension}.  In this case we say that
$\mu_\Lambda$ is {\it extendible}.

For most of this paper the lattice $\bbl$ will be taken to be $\bbz^d$ for
some $d\ge1$, and the symmetries of interest will be translations or,
occasionally, translations plus reflections about a coordinate axis.  In
this case we will refer to an invariant extension $\mu$ as a {\it TI
extension} or, as above, simply as an {\it extension}.  In \csect{trees} we
discuss the case in which $\bbl$ is the Bethe lattice.

A. G. Schlijper \cite{Schlijper1}--\cite{Schlijper4} studied the problem
of TI extensions on $\bbz^d$ in the 1980's, in connection with its
application to the {\it cluster-variation method} \cite{Pel}, used to obtain
approximations to the equilibrium thermodynamic properties of lattice
systems.  The problem was also studied by Pivato \cite{Pivato} in 2001
and by Chazottes {\it et al.} \cite{CGHU} in 2012.  For completeness we
will summarize or rederive some of the results of these papers in the
course of describing our own work; we discuss this in more detail at the
end of this section.

For $\Lambda\subset\bbz^d$ we say that measure $\mu_\Lambda$ on
$X_\Lambda$ is {\it locally translation invariant} (LTI) if for any
subsets $A,A'\subset\Lambda$, with $A'$ a translate of $A$, the marginal
of $\mu_\Lambda$ on $X_{A'}$ is the translate of the marginal on $X_A$.
The LTI property is clearly necessary for the existence of a TI extension
$\mu$ of $\mu_\Lambda$ and one question we investigate is whether or not
it is also sufficient.

 From now on we will assume that, unless otherwise stated, $\Lambda$ is
finite.  Our primary focus will be the case in which $\Omega$ is finite.
Our results then do not depend significantly on the cardinality of
$\Omega$; for simplicity of presentation we will often take $\Omega$ to
be $\{0,1\}$, so that a configuration $\eta\in X$ may be thought of as
describing a particle system or point process on the lattice.  We will
occasionally point out instances in which our results extend to certain
infinite $\Omega$.

For finite $\Omega$ we have a good understanding of the situation when
$d=1$ and $\Lambda$ is an interval.  In particular, the LTI condition is
in this case sufficient for the existence of the TI extension
\cite{Schlijper3,Pivato}, as we verify in \csect{maxent} by an explicit
construction of such an extension; the resulting $\mu$ is in fact the
unique TI extension which maximizes the entropy of its marginal on
$X_{\Lambda'}$ for any interval $\Lambda'\supset\Lambda$.  Moreover, if
$\Lambda$ contains $k+1$ sites then $\mu$ is the shift-invariant measure
on the sample paths of a $k$-step Markov process, and as such is a Gibbs
measure with TI interactions which have range at most $k$ and involve at
most $k+1$ sites.  We obtain an explicit expression for the Gibbs-Shannon
entropy of $\mu_{\Lambda'}$ as an affine function of the size of
$\Lambda'$.  This extension method applies also when $\Lambda$ is an
infinite strip in $\bbz^2$ or a higher dimensional analogue of such a
strip, or is one of a certain class of subgraphs of the Bethe lattice
which are  analogous to intervals in $\bbz$.

In the one-dimensional case, with $\Lambda$ an interval, $\Omega$ finite,
and $\mu_\Lambda$ LTI, there exist also TI extensions of $\mu_\Lambda$
which are supported on a finite number of periodic configurations
\cite{CGHU}.  We discuss this and related results in
\csects{perext}{minent}.  Every TI measure of this type has finite
entropy, in a sense which will be made precise later, and we show that
every extension with minimal entropy is of this type (there is in general
no unique minimal entropy extension).  However, for general $\mu_\Lambda$
we are not able to give an explicit construction of a minimal entropy
extension or even to determine the minimal entropy, although we do give
an {\it a priori} bound for the latter.  A key element of our analysis
here is the identification of a measure on $X_\Lambda$ with an
edge-weighting of a de~Bruijn graph; the LTI condition on the measure
then corresponds to current conservation at the vertices of the graph.

When one passes to the extension problem for $\Lambda$ not an interval of
$\bbz$ the situation is more complex.  In \csect{noextend} we give simple
examples of LTI measures $\mu_\Lambda$ which are not extendible; for
example, with $\Lambda=\{i,i+1,i+3\}\subset\bbz$ and $\Lambda$ a unit
square in $\ZZ^2$ (with four lattice sites).  For the latter example it is
relevant to remark that when $\Lambda\subset\bbz^2$ is a product of
intervals, and $\mu_\Lambda$ is LTI, the maximal entropy extension
procedure mentioned above can be used to extend $\mu_\Lambda$ in one
direction, i.e., to an infinite strip, and the resulting extension is TI in
the direction of extension.  In this process, however, the LTI property in
the second direction may be lost; only if it is preserved can one repeat the
process to obtain an extension to all of $\ZZ^2$.  Similar remarks apply in
higher dimensions.
 
 As noted above, when $\Omega=\{0,1\}$ we are studying point processes on
$\bbz^d$.  In this context the current paper may be viewed as continuing
work on a long-standing problem: the existence of point processes when one
is given partial information about the desired process (see, for example,
\cite{Percus,ST2,ST4} and further references given in \cite{KLS11}).  In
our own previous work \cite{KLS0,CKLS,KLS,KLS11} we considered the case in
which the information given was all $k$-point correlation functions of the
process up to some order $k_0$; this is in essence a {\it truncated moment
problem}.  Extending the work of \cite{AS} we obtained a sufficient
condition, essentially one of low density, for the existence of the
process.  Complete necessary and sufficient conditions were found in
\cite{KLS11}; these are worked out in detail for the case $k_0=2$, but the
method is general.  See also \cite{LM}.

We finally note that, for fixed $\Lambda\subset\bbz^d$, the set of measures
$\mu_\Lambda$ on $X_\Lambda$ which can be extended to a TI measure on
$\bbz^d$ is closed in the weak topology. For if
$(\mu_{\Lambda,n})_{n=1}^\infty$ is a sequence of measures on $X_\Lambda$
converging weakly to a measure $\mu_\Lambda$, and $\hat\mu_{n}$ is
a TI extension of $\mu_{\Lambda,n}$ to $X$, then some subsequence of
$(\hat\mu_{n})_{n=1}^\infty$ will converge, by compactness, and the
limit will be a TI extension of $\mu_\Lambda$.

The outline of the rest of the paper is as follows.  In \csect{maxent} we
show that the LTI condition is sufficient, when $\Lambda$ is an interval in
$\bbz$, for the existence of a TI extension; we do so by explicit
construction of the maximal entropy, Markovian, extension. We also point
out that this extension is Gibbsian. In \csect{trees} we take up the
extension problem on the Bethe lattice; there also we construct maximal
entropy extensions which are Markovian and Gibbsian.  In \csect{perext} we
return the the case in which $\Lambda$ is an interval in $\bbz$ and discuss
a different class of extensions, those supported on periodic
configurations, and in \csect{minent} show that these include the
extensions of minimal possible entropy.  In \csect{noextend} we give
several examples to show that, in $\bbz^d$, the LTI condition is not
sufficient for extendibility when $\Lambda$ is not an interval in $\bbz$;
we give also a corresponding example for the Bethe lattice.
\csect{decide} takes up the issue of the decidability of the extension
problem.

We finally discuss briefly the overlap of this paper with other work.  As
noted above, the Markovian extension of \csects{maxent}{sect_entropy} has
been discussed before, see in particular \cite{Schlijper3,Pivato}, although
its Gibbsian nature has not been noted.  Periodic extensions are discussed
\cite{Pivato,CGHU}, and our results in \csect{perext} partially overlap
with this work.  Our discussion of the Bethe lattice in \csect{trees} is,
we believe, new, as is the discussion of minimal entropy measures in
\csect{minent}.  Examples \ref{22} and \ref{33} of \csect{noextend} are
new.  Aspects of undecidability for the extension problem in higher
dimensions have been discussed in \cite{Schlijper4,Pivato,CGHU}, but we
give a formulation that we feel has something new to offer.

 \section{Maximal entropy extensions in one dimension\label{maxent}}

We will adapt the general notation of the introduction to the
one-dimensional case considered in this and in Sections~\ref{perext} and
\ref{minent} by letting $\Lambda_k$, for $k\ge0$, be the subset
$\{0,1,\ldots,k\}$ of $\ZZ$, $X_k=\{0,1\}^{\Lambda_k}$ be the set of
particle configurations on $\Lambda_k$, and $\pi_k\mu$, for $\mu$ a
measure on $X$ or on $X_j$ with $j>k$, be the marginal of $\mu$ on $X_k$.
We suppose that we are given a LTI measure $\mu_k$ on $X_k$ and construct
explicitly a (TI) measure $\mu$ on $X=\{0,1\}^{\ZZ}$ which extends
$\mu_k$, i.e., satisfies $\pi_k\mu=\mu_k$.

We will usually invoke the condition that $\mu_k$ is LTI as the requirement
that the marginal $\mu_{k-1}=\pi_{k-1}\mu_k$ of $\mu_k$ on $X_{k-1}$ be the
translation of the marginal on configurations on the sites
$\{1,2,\ldots,k\}$:
 \begin{eqnarray}
\mu_{k-1}(\eta_0,\ldots,\eta_{k-1})
  &=&\sum_{\sigma=0,1}\mu_k(\eta_0,\ldots,\eta_{k-1},\sigma),\nonumber\\
  &=&\sum_{\sigma=0,1}\mu_k(\sigma,\eta_0,\ldots,\eta_{k-1})\label{pti}
 \end{eqnarray}
 We say that $\mu_k$ is {\it symmetric} if it is invariant under
reflections about the center of $\Lambda_k$.  It is easy to verify that if
$\mu_k$ is LTI and $j<k$ then the marginal of $\mu_k$ on $X_j$ is also LTI,
and if $\mu_k$ is also symmetric then so is this marginal.  In fact, a
symmetric measure $\mu_k$ is LTI if and only if its marginal on $X_{k-1}$
is symmetric. 

 We begin the construction of $\mu$ by extending $\mu_k$ to a LTI measure
$\mu_{k+1}$ on $X_{k+1}$, defining
$\mu_{k+1}(\eta_0,\eta_1,\ldots,\eta_{k+1})=0$ if
$\mu_{k-1}(\eta_1,\ldots,\eta_k)=0$, and
 \be\label{extend}
\mu_{k+1}(\eta_0,\eta_1,\ldots,\eta_{k+1}) 
 = \frac{\mu_k(\eta_0,\eta_1,\ldots,\eta_k)
    \mu_k(\eta_1,\eta_2,\ldots,\eta_{k+1})}
   {\mu_{k-1}(\eta_1,\ldots,\eta_{k})},
 \ee
 otherwise.  One verifies easily that $\mu_{k+1}$ is LTI and has marginal
$\mu_k$ on $X_k$.  Moreover, if $\mu_k$ is symmetric, so is $\mu_{k+1}$.
 This leads immediately to the main result of this section.

\begin{theorem} \label{jll} If $\mu_k$ is a LTI measure on $X_k$ then there
exists a TI measure $\mu$ on $X$ which extends $\mu_k$.
Moreover, if $\mu_k$ is symmetric then $\mu$ may be taken to be symmetric
under any reflection.  \end{theorem}

\begin{proof} Suppose $\Lambda=\{-l,-l+1,\ldots, j\}$ with $l\ge0$ and
$j\ge k$.  By the repeated application of \eqref{extend} (and a
translation) one may find, for any $l$ and $j$, a LTI measure
$\mu_\Lambda$ on $\{0,1\}^\Lambda$ whose marginal on $X_k$ is $\mu_k$.
The result then follows from the Kolmogorov existence theorem.  The
preservation of symmetry is immediate.  \end{proof}

There is another approach to this extension procedure.  Equation
\eqref{extend} can be written as
 \be\label{extend2}
\mu_{k+1}(\eta_0,\eta_1,\ldots,\eta_{k+1}) 
 = \mu_k(\eta_0,\eta_1,\ldots,\eta_k)
   \mu_k(\eta_{k+1}\mid\eta_1,\ldots\eta_{k}),
 \ee
  where the conditional probability
$\mu_k(\zeta_k\mid\zeta_0,\ldots\zeta_{k-1})$ is defined to be zero
whenever $\mu_{k-1}(\zeta_0,\ldots\zeta_{k-1})=0$, and then one sees that 
the extension procedure described in the 
proof of \cthm{jll} is equivalent to 
 \be\label{MC}
\mu_j(\eta_0,\eta_1,\ldots,\eta_j)
  = \mu_k(\eta_0,\ldots,\eta_k)
        \prod_{i=1}^{j-k}\mu_k(\eta_{i+k}\mid\eta_i,\ldots,\eta_{i+k-1}).
 \ee
 Equation \eqref{MC} says that we may regard the extension procedure as
defining a Markov chain with state space $\{0,1\}$ having memory, for which
the transition probabilities depend on states at the previous $k-1$
time steps.  The TI extension $\mu$ of $\mu_k$ given in \cthm{jll} is then
just the invariant measure on sample paths for this chain. Equivalently one
may think of a 1-step Markov chain with state space $X_k$.

\begin{remark}\label{spins} It is sometimes convenient to use spin
notation rather than particle (i.e., lattice gas) notation in describing
configurations in $X_k$ or $X$.  As usual if $\eta$ is a particle
variable taking values in $\{0,1\}$ we introduce a corresponding spin
variable $\sigma=2\eta-1$ taking values in $\{-1,+1\}$.  This is
convenient in particular because it permits us to write a measure $\mu_k$
directly in terms of the corresponding spin correlations:
 \be\label{sform}
\mu_k(\sigma_0,\ldots,\sigma_k)
  =2^{-(k+1)}\biggl(1+\sum_{A\subset\Lambda_k,A\ne\emptyset}
              \langle\sigma_A\rangle\,\sigma_A\biggr),
 \ee
 where as usual $\sigma_A=\prod_{i\in A}\sigma_i$ and the spin
correlation $\langle\sigma_A\rangle$ denotes the expectation
$\mu_k(\sigma_A)$ of $\sigma_A$ in the measure $\mu_k$.  One may think of
the expectations $\langle\sigma_A\rangle$ in \eqref{sform} as parameters
which determine the measure; the condition that $\mu_k$ be LTI is simply
that $\langle\sigma_A\rangle=\langle\sigma_B\rangle$ whenever $B$ is a
translate of $A$.  Note, however, that in using the form \eqref{sform} to
construct a measure with certain given $\langle\sigma_A\rangle$ one must
check that it assigns nonnegative probability to each configuration.
\end{remark}

\begin{remark}\label{Omega} (a) If $\Omega$ is any finite set then
everything said above extends immediately to measures defined on
$\Omega^{\Lambda_k}$.  As an example, take $\Omega=\{0,1\}^{\Lambda_m}$
so that $\Omega^{\Lambda_k}$ may be identified with the space
$\{0,1\}^{\Lambda_k\times\Lambda_m}$ of particle configurations on a
$k\times m$ rectangle.  If such a measure is LTI under translations in
the horizontal direction then it may be extended by an obvious
generalization of \eqref{MC} to a measure on configurations on
$\Lambda\times\Lambda_m$, where $\Lambda$ is any interval containing
$\Lambda_k$, and hence to a measure on configurations on
$\ZZ\times\Lambda_m$ as in the proof of \cthm{jll}.  However, even if the
original measure is also LTI in the vertical direction, this extension
procedure need not maintain the LTI property for translations in the
vertical direction.  \cex{22} (see \csect{noextend}) is of this nature:
an extension from the original $2\times2$ square to a $3\times2$
rectangle destroys the vertical LTI property.

 \smallskip\noindent
 (b) The extension method in the Markovian form \eqref{MC} is
well defined for a large class of spaces $\Omega$; it suffices that the
$\mu_k(\eta_k\mid\eta_0,\ldots\eta_{k-1})$ be well defined (regular)
conditional probabilities.  This will be true, for example, if $\Omega$ is
a complete separable metric space \cite{PS} or, essentially equivalently, a
{\it standard Borel space} \cite{RBS}.  Consider, for example, a LTI
measure for a strip in $\bbz^2$, that is, a LTI measure on
$\{0,1\}^{\ZZ\times\Lambda_m}=\left(\{0,1\}^{\ZZ}\right)^{\Lambda_m}$;
since $\{0,1\}^{\ZZ}$ is a standard Borel space we may then use \eqref{MC}
to extend vertically to a measure on configurations on $\{0,1\}^{\bbz^2}$.
In this context the original measure might be obtained by beginning with a
LTI measure on a rectangle $\{0,1\}^{\Lambda_k\times\Lambda_m}$ and making
the maximal entropy extension in the horizontal direction.  As noted above,
in general the LTI property in the vertical direction will not be preserved
by this extension, but if it is, the above remarks apply.  \end{remark}

 \subsection{Entropy\label{sect_entropy}}

 Let us denote the Gibbs-Shannon entropy of a probability 
measure $\mu$ on a discrete set $Q$ (typically for us a set of
configurations) by $S(\mu)$:
 \be\label{entropy}
  S(\mu)=-\sum_{q\in Q}\mu(q)\log\mu(q).
 \ee
 The entropy of the extension $\mu_j$ given in \eqref{extend2} or
 \eqref{MC} is easily found to be 
 \be\label{ent}
 S(\mu_j)=S(\mu_k) + (j-k)[S(\mu_k)-S(\mu_{k-1})],\qquad j\ge k.
 \ee
 Now recall the strong subadditivity inequality for entropy \cite{L75}:
if $\mu$ is a probability measure on a product
$Q=Q_1\times Q_2\times Q_3$ then the marginal measures $\mu_{12}$,
$\mu_{23}$, and $\mu_2$ on $Q_1\times Q_2$, $Q_2\times Q_3$, and $Q_2$,
respectively, satisfy
 \be
S(\mu)\le S(\mu_{12})+S(\mu_{23})-S(\mu_2).
 \ee
  This implies that if $\nu_j$ is any measure on $X_j$ for which the
marginal $\nu_j^{(i)}$ on ${\{0,1\}^{\{i,i+1,\ldots,i+k\}}}$ satisfies
 \be\label{bpti}
\nu_j^{(i)}(\eta_i,\ldots,\eta_{i+k})
   =\mu_k(\eta_i,\ldots,\eta_{i+k}), \qquad i=0,1,\ldots j-k,
 \ee
 then $S(\nu_j)$ is at most equal to the right hand side of \eqref{ent}.
Thus $\mu_j$ is in fact the maximal entropy extension of $\mu_k$ to $X_j$
under the constraint \eqref{bpti}; there is a unique such extension
by the strict concavity of the entropy.  Note that satisfaction of
\eqref{bpti} is necessary for $\nu_j$ to be LTI but is not sufficient,
since it does not impose translation invariance for sets of diameter
greater than $k+1$.  Nevertheless, we know from \cthm{jll} that the measure
$\mu_j$, which maximizes the entropy under only the constraint
\eqref{bpti}, is LTI, so that the same measure would be obtained by
maximizing the entropy over all LTI measures.

 \subsection{Gibbs measures\label{gibbs}}

Now suppose that $\mu_k$ does not assign zero probability to any
configuration.  Then it is easy to see that the Markovian extension $\mu$
on $X$ provided by the proof of Theorem 1 is an infinite volume
Gibbs state (see \cite{Ruelle} for definitions and results) for the
translation invariant interaction energy $U$ defined formally by
 \be
 U(\eta)=\sum_{i\in\bbz}\Phi(\eta_i,\ldots,\eta_{i+k}),\quad\eta\in X,
 \ee
 where the interaction potential $\Phi$ is given by
 \be\label{Phi}
\Phi(\eta_0,\ldots,\eta_k)=-\ln\mu_k(\eta_k\mid\eta_0,\ldots,\eta_{k-1}).
 \ee
If $\mu_k$ assigns zero measure to some configurations then the same
arguments and conclusions apply. We may either simply allow
for the interaction $\Phi$ to sometimes be $+\infty$ and consider the
corresponding Gibbs states, or we may follow \cite{Ruelle} and restrict
the allowable configurations to those whose restriction to any interval
of $k+1$ sites has nonzero probability under $\mu_k$.

Conversely, suppose that one starts with a TI measure $\nu$ on $X$ which is
a Gibbs state for some interaction potential with interactions of range at
most $k$; again we may allow the interaction potential to take the value
$+\infty$ or alternatively may restrict the allowable configurations by
local constraints of range $k$.  The marginal $\nu_k$ of $\nu$ on $X_k$
will then be LTI.  Since $\nu$ is Gibbs, it is Markovian of range $k$
\cite{Ruelle} and hence agrees with the Markovian extension of $\nu_k$ to
$X$ provided by the proof of Theorem 1.  Note that the original interaction
in terms of which $\nu$ is defined need not agree with the interaction
\eqref{Phi} naturally associated the Markovian extension procedure applied
to $\nu_k$; the same Gibbs measure will arise from different interactions
if they give the same answer, up to a constant, when they are used to
compute the energy in a finite volume with boundary condition (the constant
may depend on the boundary condition) \cite[Section 2.4]{Georgii}.

\section{Extension of measures on the Bethe lattice\label{trees}}

As discussed in \crem{Omega}(a) above, the Markovian extension technique of
\csect{maxent} does not in general yield extensions of measures in higher
dimensions.  It does do so, however, on the Bethe lattice.  To discuss this
we fix $q\ge1$ and let $\bbt$ denote the (infinite) Bethe lattice in which
every vertex has degree $q+1$.  For a subgraph $H$ of $\bbt$ we may then
consider a measure $\mu_H$ on $X_H=\Omega^H$ and ask for an extension to
$X=\Omega^{\bbt}$ which is invariant under the group of all isomorphisms of
$\bbt$.  Again there is an obviously necessary {\it local invariance (LI)}
condition for the existence of such an extension: that if $H_1$ and $H_2$
are connected isomorphic subgraphs of $H$, with $\phi:H_1\to H_2$ an
isomorphism, and $\nu_j$ is the marginal of $\mu_H$ on $X_{H_j}$, then the
induced mapping $\phi_*$ of measures satisfies $\phi_*\nu_2=\nu_1$.

\def\mini#1#2#3{\beginpicture\setcoordinatesystem units <10pt,10pt> 
\setplotarea x from -0.9 to 0.9, y from -0.2 to 0.7 
\put{$\scriptstyle#1$} at -0.866 0.5 \put{$\scriptstyle#2$} at 0 0 
   \put{$\scriptstyle#3$} at 0.866 0.5
\plot -0.6928 0.4 -0.1732 0.1 / \plot 0.1732 0.1 0.6928 0.4 /
\endpicture}

 \begin{figure}[ht!]
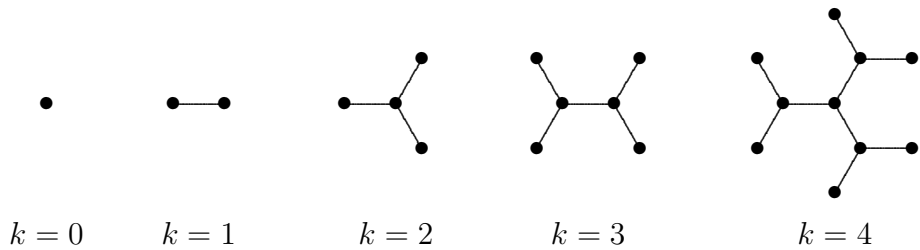
  \input ptifig_trees2
 \bigskip
 \caption{Full subgraphs of diameter $k$, $k=0,\ldots,4$, when
$q=2$.\label{treefig}} \end{figure}

A special role in what follows will be played by {\it full subgraphs} of
$\bbt$: the maximal subgraphs of with fixed diameter $k$.  Any two such
are isomorphic (see \cfig{treefig}).  When $k$ is even, such a full
subgraph  has $(q+1)$-fold symmetry about a central vertex; when $k$
is odd, it has 2-fold symmetry across a central edge.  We fix some vertex
$v_0$ in $\bbt$ and some edge $e_0$ incident on it, and for
$k=0,1,\ldots$ let $T_k$ denote the particular full subgraph 
which has central vertex $v_0$ if $k$ is even and central edge $e_0$ if
$k$ is odd, so that in particular $T_0\subset T_1\subset \cdots$.  We
abbreviate $X_{T_k}$ as $X_k$.  If $q=1$ then the Bethe lattice $\bbt$ is
just $\bbz$ and $T_k$ is the interval $\Lambda_k$ of $k+1$ sites.  For
general $q$ the graphs $T_k$ play a role for extensions to $\bbt$ similar
to that of the $\Lambda_k$ for extension to $\bbz$: when one begins with
an LI measure on $X_k$ there is no obstruction to extension.

\begin{theorem}\label{treeext} {\sl Let $\mu_k$ be an LI measure on $T_k$,
$k=0,1,\ldots$.  Then $\mu_k$ has an extension to an invariant measure
$\mu$ on $X$.} \end{theorem}

  Note that, because the invariance we consider for $\bbt$ is invariance
under all isomorphisms, the $q=1$ case of \cthm{treeext} corresponds to the
case in \cthm{jll} of symmetric extensions of symmetric measures.

\begin{proofof}{\cthm{treeext}} We first extend $\mu_k$ to a measure
$\mu_{k+1}$ on $T_{k+1}$; in doing so we will need only two consequences of
the local invariance of $\mu_k$: (i)~$\mu_k$ is invariant under
automorphisms of $T_k$, and (ii)~the marginal of $\mu_k$ on $X_{k-1}$ is
invariant under automorphisms of $T_{k-1}$. (In fact, (i) and (ii) imply
local invariance, but we do not need this.)  $T_{k+1}$ contains $r_k$
copies $T_k^1,\ldots,T_k^{r_k}$ of $T_k$ (one of which is $T_k$ itself),
where $r_k=2$ if $k$ is even and $r_k=q+1$ if $k$ is odd, and these
intersect in $T_{k-1}$.  From (i) it follows that we may speak without
ambiguity of the copy $\mu_k^i$ of $\mu_k$ on $T_k^i$, and from (ii) that
the marginal of any $\mu_k^i$ on $T_{k-1}$ is independent of $i$.  Then the
extension formula is
 \begin{align}
 \mu_{k+1}(\eta)&=\mu_k(\eta_{T_{k-1}})^{-(r_k-1)}
    \prod_{i=1}^{r_k}\mu_k^i(\eta_{T_k^{i}})\label{ext1}\\
   &=\mu_k(\eta_{T_{k-1}})
    \prod_{i=1}^{r_k}\mu_k^i(\eta_{T_k^{i}\setminus T_{k-1}} 
        \mid \eta_{T_{k-1}}).\label{ext2}    
 \end{align}
 where for $\eta\in X_{k+1}$ we have written $\eta_H$ for the restriction
of $\eta$ to any subgraph $H$ of $T_{k+1}$.  Note that \eqref{ext1} is a
direct generalization of \eqref{extend}, while \eqref{ext2} is a
symmetrized version of \eqref{extend2}.

 It is easy to see that $\mu_{k+1}$ is indeed an extension of $\mu_k$.
Moreover, it again satisfies properties (i) and (ii); (i) follows from the
fact that any automorphism of $T_{k+1}$ will permute the $T^i_k$ and
preserve $T_{k-1}$, and because $\mu_{k+1}$ extends $\mu_k$, (ii) is an
immediate consequence of property (i) for $\mu_k$.  Thus we may apply the
above construction repeatedly to extend $\mu_k$ to any $X_j$ with $j>k$,
maintaining properties (i) and (ii) at each step, and thus obtain the
extension $\mu$ on $X$ as a limit.  $\mu$ is invariant under the group $G$
of isomorphisms of $\bbt$, because, as is easily seen, $G$ is generated by
the subgroup $G_{v_0}$ of $G$ of isomorphisms that fix the central vertex
$v_0$, together with the subgroup $G_{e_0}$ of isomorphisms that fix the
central edge $e_0$, and invariance of $\mu$ under these subgroups follows
from properties (i) and (ii) of the finite extensions.

\end{proofof}

\begin{remark}\label{maxmar} There is an alternative but equivalent
version of the last step of this proof, that is, of the extension beyond
$X_{k+1}$.  The procedure is of course again inductive; we consider the
extension from $X_j$ to $X_{j+1}$, $j\ge k+1$.  One obtains $T_{j+1}$
from $T_j$ by adding $q$ new edges and vertices to each vertex $v$ in a
certain subset $V_j$ of the vertices of $T_j$ of degree one.  Each such
$v$, together with the new edges and vertices attached to it, is part of
a unique copy $\widehat T_k^{v}$ of $T_k$ in $T_{j+1}$, and the extension
is
 \be
\mu_{j+1}(\eta)=\mu_j(\eta_{T_j})
   \prod_{v\in V_j}\mu_k(\eta_{\widehat T_k^{v}\setminus H^{v}} 
    \mid \eta_{H^{v}}),
 \ee
 where $H^{v}=\widehat T_k^{v}\cap T_j$. Thus the distribution determined
by $\mu_{j+1}$ on the new vertices attached to $v$ depends only on the
configuration in $H^v$.  From this it follows that the distribution on
the complement of any $T_j$ which arises from  conditioning the
extension $\mu$ on the configuration in $T_j$, actually depends only on
the configuration on vertices of $T_j$ a distance at most $k$ from the
complement.  In this sense the extension constructed here is maximally
Markovian.  \end{remark}

It is straightforward to calculate the entropy of these extensions, through
the recurrence 
 \be
S(\mu_{j+1})=S(\mu_j)+(r_j-1)[S(\mu_j)-S(\mu_{j-1})],
 \ee
 where $\mu_{k-1}$ is the marginal of $\mu_k$ on $X_{k-1}$. For $q=1$ the
entropy grows linearly with $q$; see \csect{sect_entropy}.  For $q\ge2$,
$S(\mu_j)$ grows as $q^{j/2}$ for large $j$; more specifically, with
$\Delta S=[S(\mu_k)-S(\mu_{k-1})]$ we find that for $j\ge k$ and
$p=\lfloor(j-k+1-(k\bmod2))/2\rfloor$,
 \be\label{bent}
  S(\mu_j)=S(\mu_k)+\frac{r_{j+1}q^p-r_{k+1}}{q-1}\,q^{k\bmod2}\,\Delta S.
 \ee
  A simple calculation shows that the entropy per site is asymptotically
constant, and an argument as in \csect{sect_entropy} implies that
\eqref{bent} is the maximum possible entropy of any extension.

The considerations of \csect{gibbs} apply here essentially unchanged; in
particular, the measure on $\bbt$ constructed here is Gibbsian with
interaction potential
 \be
\Phi(\eta_{T_k})=-\ln\mu_k(\eta_{T_k\setminus T_{k-1}} 
        \mid \eta_{T_{k-1}}).
 \ee

\section{Periodic extensions in one dimension \label{perext}}

We now turn again to the study of TI extensions to $X=\Omega^{\bbz}$ of an
LTI measure $\mu_k$ defined on the space $X_k$ of configurations on the
interval $\Lambda_k$. We consider in particular extensions which are of a
different character from the maximal entropy extensions considered in
\csect{maxent}: those supported on periodic configurations in $X$.  We will
call such a measure a {\it periodic configuration} (PC) measure.  Within
the class of PC measures we single out the {\it basic periodic
configuration} (BPC) measures: a BPC measure is obtained from a particular 
$p$-periodic configuration by giving weight $1/p$ to each
of its $p$ translates.  The set
of PC measures is clearly convex and its extreme points are the BPC
measures; since the set of BPC measures is countable, every PC measure
$\nu$ has a representation
 \be\label{PC}
\nu=\sum c_\alpha\nu_\alpha,\qquad c_\alpha\ge0,\quad\sum c_\alpha=1,
 \ee
 where the sums run over the set $\{\nu_\alpha\}$ of all BPC measures. 

In this section and the next we will for simplicity consider only the case
of particle configurations, i.e., $\Omega=\{0,1\}$; the extension of the
results to any finite $\Omega$ is straightforward.  We will write a
configuration $\eta\in X_k$ as a string $\eta_0\eta_1\ldots\eta_k$ (rather
than as a $(k+1)$-tuple $(\eta_0,\ldots,\eta_k)$) and will similarly write
$\eta\in X$ as a doubly infinite string; since we do not use any products
of the $\eta_i$ no confusion should arise.

 \subsection{BPC  and extremal LTI measures on
$X_k$\label{structure}}

Let $M_k$ denote the set of all LTI probability measures on $X_k$; $M_k$ is
manifestly a convex polytope.  Since the LTI condition \eqref{pti} and the
normalization $\sum_{\eta\in X_k}\mu_k(\eta)=1$ represent $2^k$ independent
constraints on the $2^{k+1}$ variables $\mu_k(\eta)$, $\eta\in X_k$, $M_k$
has dimension $2^k$.  We wish to characterize the extreme points of this
polytope.  We will say that a periodic configuration $\eta\in X$ and the
corresponding BPC measure $\nu$ supported on the translates of $\eta$ are
{\it $k$-primitive} if no segment of length $k$ is repeated within one
(minimal) period $p$ of $\eta$, that is, if
$\eta_i\eta_{i+1}\cdots\eta_{i+k-1}=\eta_j\eta_{j+1}\cdots\eta_{j+k-1}$
only when $j-i$ is a multiple of $p$.  Note that then $p\le2^k$.  We will
prove shortly, using the concept of a {\it de Bruijn graph}, that the
extreme points of $M_k$ are precisely the marginals $\pi_k\nu$ with $\nu$ a
$k$-primitive BPC measure.

The (binary) {\it De Bruijn graph of order $k\ge1$} (or sometimes {\it
dimension $k$}) \cite{deBruijn,Good}, $G_k$, is the directed graph with $2^k$
vertices and $2^{k+1}$ edges, labeled respectively by all binary strings of
length $k$ and length $k+1$, in which for any binary digits $a$ and $b$ and
binary string $\theta$ of length $k-1$ the edge $a\theta b$ runs from
vertex $a\theta$ to vertex $\theta b$ (see \cfig{dbg}).  Note in particular
that $G_k$ contains two loops, labeled respectively by $00\cdots0$ and
$11\cdots1$, but no multiple edges.  Since the edges of $G_k$ are labeled
by the elements of $X_k$, it is clear that any probability measure $\mu_k$
on $X_k$ corresponds to an assignment of a nonnegative {\it current}
$j_\eta$ to each edge $\eta$ of $G_k$ in such a way that
$\sum_\eta j_\eta=1$; the correspondence is of course via
$j_\eta=\mu_k(\{\eta\})$.  The terminology ``current'' is appropriate
because one checks easily that $\mu_k$ is LTI if and only if the current is
conserved at each vertex of $G_k$, that is, if for each vertex $\xi$ the
sum of the currents on the two edges entering $\xi$ is equal to the sum of
the currents on the two edges leaving $\xi$.  (Such a current assignment is
called a {\it circulation} in the theory of graphs.)

\bigskip
\begin{figure}[ht!]
\hskip20pt\includegraphics[width=9.2cm]{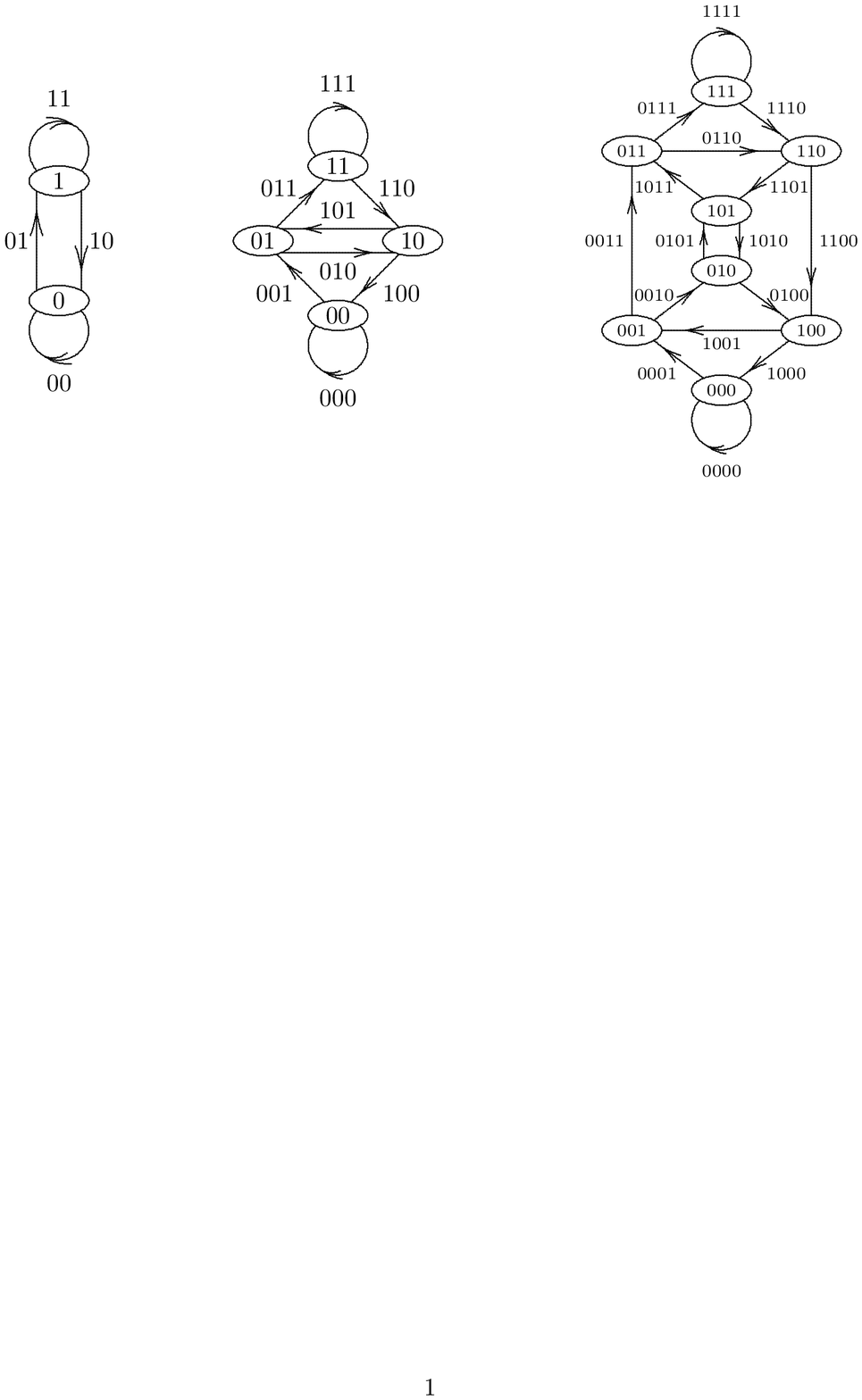}
\caption{The De Bruijn graphs $G_1$, $G_2$, and $G_3$.\label{dbg}}
\end{figure}

 Suppose that $\P$ is a closed path in $G_k$, that is, a sequence of $|\P|$
edges $\eta^{(1)},\ldots,\eta^{(|\P|)}$ in that order (possibly with
repetitions), and that $\P$ is {\it minimal} in the sense that there is no
shorter closed path $\P'$ such that $\P$ is obtained by tracing $P'$
several times.  Each $\eta^{(j)}$ is a string of $k+1$ symbols (see
\cfig{dbg}), and we may obtain from $\P$ a periodic configuration by
tracing $\P$ repeatedly, recording the first symbol on each edge
encountered: with $\eta^{(j)}=\eta^{(j)}_0\cdots\eta^{(j)}_k$ this
configuration is
$\cdots\eta^{(|\P|)}_0\eta^{(1)}_0\cdots\eta^{(|\P|)}_0\eta^{(1)}_0\cdots\in
X$.  Let $\nu_\P$ denote the BPC measure associated with this
configuration; the map $\P\mapsto\nu_\P$ establishes a bijective
correspondence between minimal closed paths and BPC measures.  As usual a
{\it cycle} $\C$ in $G_k$ is a closed path with no repeated vertices; a
loop is considered to be a cycle.

\begin{theorem}\label{measures} The mapping $\C\to\nu_\C$ gives a bijective
correspondence between the cycles in $G_k$ and the $k$-primitive BPC
measures on $X$.  The extreme points of $M_k$ are precisely the marginals
$\pi_k\nu_\C$, with $\C$ a cycle in $G_k$; equivalently, every LTI measure
$\mu_k$ on $X_k$ has the form $\pi_k\nu$, where $\nu$ is a (finite) convex
combination of the $k$-primitive BPC measures $\nu_\C$.  In particular,
every such $\mu_k$ has an extension to a PC measure.  \end{theorem}

\begin{proof} The condition that a cycle $\C$ contain no repeated
vertices is precisely the condition of $k$-primitivity on $\nu_\C$, and
the bijective nature of the correspondence $\C\leftrightarrow\nu_\C$ is
clear.  If $|\C|=p$ then the measure $\pi_k\nu_\C$ gives weight $1/p$ to
each edge of $\C$ and weight $0$ to all other edges, and is thus clearly
an extreme point of $M_k$.  Conversely, given an element $\mu_k\in M_k$,
the set of edges $\eta$ of $G_k$ for which $\mu_k(\eta)>0$ must by
current conservation contain a cycle $\C$, say of length $p$, and if
$\alpha$ is the smallest current assigned by $\mu_k$ to any edge of $\C$
then we may write
 \be\label{mustar}
\mu_k=(p\alpha)\pi_k\nu_\C+(1-p\alpha)\mu_k^*,\qquad\hbox{where}\quad
   \mu_k^*=\frac1{1-p\alpha}(\mu_k-p\alpha\pi_k\nu_\C),
 \ee
 thus representing $\mu_k$ as a convex combination of the LTI probability
measures $\nu_\C$ and $\mu_k^*$.  We may apply the same construction to
$\mu_k^*$ and by repeating this process  obtain a representation of
$\mu_k$ as a convex combination
 \be\label{cc}
\mu_k=\sum_\C\beta_\C\pi_k\nu_\C,\qquad \sum_\C\beta_\C=1,
 \ee
 as required; the process must terminate because $\mu_k^*$ as
defined by \eqref{mustar} gives zero probability to at least one more
configuration than does $\mu$.  \end{proof}

\begin{remark}\label{cycles} (a) It is known \cite{SG,M} that $G_k$ contains
cycles with every length $p$ satisfying $1\le p\le2^k$.  In particular,
$G_k$ contains two cycles of length 1, the loops mentioned above, which
correspond to the two period-one configurations in $X$, and
$2^{(2^{k-1}-k)}$ cycles of length $2^k$ \cite{deBruijn}.

 \smallskip\noindent
 (b) Although it follows from (a) that $M_k$ contains at least
$2^{(2^{k-1}-k)}$ extreme points, a result of Carath\'eodory \cite{G,C} implies
that, because $M_k$ is of dimension $2^k$, every LTI measure $\mu_k$ can
be written as a convex combination of at most $2^k+1$ measures
$\pi_k\nu_\C$.

  \end{remark}

 \subsection {Extensions of LTI measures to a ring\label{torext}}

Let $\bbz_L$ denote the ring (one-dimensional discrete torus) with $L$ sites:
$\bbz_L=\{0,\ldots,L-1\}$, with translation defined via periodic boundary
conditions.  Rather than considering extensions of $\mu_k$ to TI measures
on $X$ one may instead ask whether there exists an integer $L\ge k+1$
such that $\mu_k$ has an extension to a TI measure on
$Y_L=\{0,1\}^{\ZZ_L}$, the set of configurations on $\ZZ_L$, and, if so,
for what values of $L$ this is true.  The first question is clearly
equivalent to that of whether there is an extension to a measure $\nu$ on
$\ZZ$ which is a convex combination of a finite number of BPC measures;
if so, $L$ may then be taken to be the least common multiple of their
periods.  \cthm{measures} implies that the answer is affirmative; here
we discuss a few questions about the values of $L$.

\begin{theorem} \label{torus} The smallest integer $L\ge k+1$ such that
every LTI measure on $X_k$ may be extended to a TI measure on $Y_L$ is
 \be\label{lcm}
L_0(k)=\lcm\{1,2,3,\ldots,2^k-1, 2^k\}.
 \ee
  \end{theorem}

\begin{proof} Suppose that $L\ge k+1$. It is easy to see that the measure
$\nu_\C$, with $\C$ a cycle in $G_k$, can be extended to a TI measure on
$Y_L$ if and only if $|\C|$ divides $L$, and from the representation
\eqref{cc} it follows that every LTI measure on $X_k$ can be extended if
and only if $L$ is divisible by $|\C|$ for every cycle $\C$.  The result
\eqref{lcm} now follows from \crem{cycles}(a).  \end{proof}

 \cthm{torus} does not address the natural question of the minimal value of
$L$, or of all possible values of $L$, for the extension of a given
measure $\mu_k$.  We will not try to give a general discussion of
this matter, but rather content ourselves with discussing two examples.

\begin{example}\label{ktwo} Suppose that we begin with a measure $\mu_1$
defined on two sites.  The De Bruijn graph $G_1$ has three cycles
$\C_0,\C_1,\C_2$ (see \cfig{dbg}), where $\C_0$ and $\C_1$ denote the
loops $00$ and $11$ respectively, and $\C_2$ the cycle with edges $01$
and $10$, so that $\mu_1=\sum_{i=0}^2c_i\pi_1\nu_{\C_i}$ with
coefficients in the simplex
$\Sigma=\{(c_0,c_1,c_2)\mid c_i\ge0,\ \sum_{i=0}^2c_i=1\}$. $\mu_1$ has a
TI extension to $Y_L$ for any even $L$, since the measures corresponding
to the three extreme points $(1,0,0)$, $(0,1,0)$, and $(0,0,1)$ of $\Sigma$
have extensions which give respectively probability one to $0^L$,
probability one to $1^L$, and probability one half to each of
$(01)^{L/2}$ and $(10)^{L/2}$ (where for any finite binary string $\eta$,
$\eta^n$ denotes the concatenation of $n$ copies of $\eta$).  If $L$ is
odd, say $L=2\ell+1$, then $\mu_1$ has an extension to $Y_L$ if and only
if $c_2\le1-1/L$.  For the measures corresponding to the four extreme points
$(1,0,0)$, $(0,1,0)$, $(1/L,0,1-1/L)$, and $(0,1/L,1-1/L)$ of
$\Sigma\cap\{(c_0,c_1,c_2)\mid c_2\le1-1/L\}$ all have extensions to
$Y_L$; the first two are as above and the remaining two give probability
$1/L$ to each of the $L$ translates of $0(01)^\ell$ and $1(01)^\ell$,
respectively.  Conversely, if $L$ is odd and $\mu_1$ has a TI extension
to $Y_L$ then the fact that every configuration in $Y_L$ contains either
two consecutive 1's or two consecutive 0's implies that
$c_0+c_1\ge1/L$.\end{example}

\begin{example}\label{numnth} If $\mu_k$ is supported on a union of
vertex disjoint cycles (and gives positive probability to each) then an
analysis as in the proof of \cthm{torus} shows that $\mu_k$ has an
extension to $Y_L$ if and only if $L$ is divisible by the least common
multiple of the lengths of these cycles.  The maximal length of a cycle
is $2^k$ (see \crem{cycles}(a)), but this least common multiple may be
much larger.  Using Maple, we find that for $k=1,\ldots,7$ the maximum
value $L_k$ of this least common multiple, taken over all collections of
vertex disjoint cycles, is as given in Table~\ref{Lk}; there the
factorization corresponds to the lengths of the maximizing set of cycles.
In all of these cases except $k=3$, $L_k=g(2^k)$, where $g$ is {\it
Landau's function} \cite{L}: $g(n)$ is the maximum, over all partitions
$\{n_1,\ldots,n_q\}$ of $n$, of $\lcm\{n_i,\ldots,n_q\}$.  Landau showed
that  $\log g(n)\sim\sqrt{n\log n}$ for large $n$.  We do not know
whether $L_k=g(2^k)$ for some or all $k\ge8$; the question is whether or
not there exists a set of vertex-disjoint cycles in $G_k$ whose lengths
are given by the maximizing partition of $2^k$.

\begin{table}[ht!]
\begin{center}\begin{tabular}{|r|r@{}c@{$\;=\;$}l|}
\hline
$k$&\multicolumn{3}{|c|}{$L_k$}\\
\hline
1&2&&2\\
2&4&&4\\
3&12&&$3\cdot4$\\
4&140&&$4\cdot5\cdot7$\\
5&5460&&$3\cdot4\cdot5\cdot7\cdot13$\\
6&2042040&&$3\cdot5\cdot7\cdot8\cdot11\cdot13\cdot17$\\
7&7216569360&&$5\cdot7\cdot9\cdot11\cdot13\cdot16\cdot17\cdot19\cdot31$\\
\hline
\end{tabular}\end{center}
\caption{Maximum values of $\lcm\{|\C_1|,\ldots,|\C_n|\}$ for
  $\C_1,\ldots,\C_n$ vertex disjoint cycles in $G_k$.
\label{Lk}}
\end{table}

\end{example}

 \section{Minimal entropy extensions\label{minent}}

In this section we consider TI extensions $\mu$ of $\mu_k$ with finite
total entropy, where the {\it total entropy} of any TI measure $\mu $ on
$X$ is defined by
 \be\label{edef}
S(\mu)=\lim_{j\to\infty}S(\pi_j\mu),
 \ee
 with the entropy $S(\pi_j\mu)$ given by \eqref{entropy}.  The TI extensions
 described in \cthm{measures} clearly have finite entropy.  We begin with a
 partial converse.

\begin{lemma}\label{shel} Every TI measure $\mu$ on $X$ with finite total
entropy is a PC measure.  \end{lemma}

\begin{proof} Let $\mu$ be a TI measure on $X$.  Since any configuration
$\eta$ with $\mu(\eta)>0$ must be periodic, $\mu$ is atomic if and only if
it is a PC measure.  Suppose then that $\mu$ is not atomic, so that
$\mu=t\mu_a+(1-t)\mu_b$ with $0\le t<1$, $\mu_a$ an atomic probability
measure, and $\mu_b$ a nonatomic probability measure.  Because $\mu_b$ is
nonatomic there must exist, for any $n\ge1$, a partition
$X=\bigcup_{i=1}^n A_i$ with the $A_i$ measurable sets satisfying
$\mu_b(A_i)=1/n$ \cite[p.174]{H}.  Then one can approximate the $A_i$
arbitrarily closely by disjoint cylinder sets $A'_i$, where for some large
$N$, $A'_i=\pi_{[-N,N]}^{-1}(B_i)$ with $B_i\subset X_{[-N,N]}$.  In
particular, one can make this approximation sufficiently close that
 \begin{eqnarray}
  S(\pi_{[-N,N]}\mu_b)\nn
  &\ge&-\sum_{i=1}^n(\pi_{[-N,N]}\mu_b)(B_i)\log(\pi_{[-N,N]}\mu_b)(B_i)\nn\\
  &\ge&-\frac12\sum_{i=1}^n\mu_b(A_i)\log\mu_b(A_i)\\
  &=&\frac12\log n.\nn
 \end{eqnarray}
 Thus $\mu_b$ and hence also $\mu$ does not have finite (total) entropy.
 \end{proof}

Now let $\mu_k$ be a LTI measure on $X_k$; from \clem{shel} and the
discussion in \csect{structure} we know that every finite entropy
extension $\nu$ of $\mu_k$ has the form
 \be\label{muhat}
 \nu=\sum_\P c_\P\nu_\P, \qquad c_\P\ge0,\quad\sum_\P c_\P=1,
 \ee
where the sum is over some collection of minimal closed paths in
$G_k$.  Note that the BPC measure $\nu_\P$ has entropy
$S(\nu_\P)=\log|\P|$ and that the measure $\nu$ of \eqref{muhat} has
entropy
 \be\label{entnu}
S(\nu)=\sum c_\P(\log|\P|-\log c_\P).
 \ee

\begin{theorem}\label{minentth} Let $\mu_k$ be a LTI measure on $X_k$.
Then $\mu_k$ has at least one (TI) extension of minimal entropy of the form
 \be\label{muhatf}
\hat\mu=\sum_{i=0}^{2^{k}}c_i\nu_{\C_i}, \qquad 
    c_i\ge0,\quad\sum_{i=0}^{2^k}c_k=1,
 \ee
 for some cycles $\C_0,\ldots,\C_{2^k}$ in $G_k$.
 \end{theorem}

 \begin{proof} We will first show that an entropy minimizing extension
exists, then that there is such a measure with the form \eqref{muhat} having
$c_\P>0$ only if $\P$ is a cycle, and finally that there is such a
minimizing measure which is a convex combination of at most $2^k+1$
measures $\nu_\C$, with $\C$ a cycle.  Let $M$ be the space of all TI
probability measures on $X$ which extend $\mu_k$. Then, in the weak
topology, $M$ is compact and moreover \eqref{edef} exhibits $S$ as the
increasing limit of functions continuous on $M$; thus $S$ is lower
semi-continuous on $M$ and so achieves its minimum there.  Each minimizing
measure will have the form \eqref{muhat}.

Now we show that there is a minimizing measure in the form \eqref{muhat}
for which each closed path appearing with nonzero coefficient is a cycle.
To do so we note that among all the minimizing measures we can choose one,
say
 \be\label{mutilde}
 \tilde\mu=\sum_\P \tilde c_\P\nu_\P, \qquad 
  \tilde c_\P\ge0,\quad\sum_\P \tilde c_\P=1,
 \ee
which contains, among the (minimal) closed paths $\P$ with $\tilde c_\P>0$,
the maximal number of cycles.  Then in fact all the $\P$ in \eqref{mutilde}
for which $\tilde c_\P>0$ must be cycles.  For if there is some such path
$\tilde\P$ which is not a cycle it may be decomposed into a cycle
$\tilde\C$ and some remaining path $\P'$ edge-disjoint from $\tilde\C$.  If
now in \eqref{mutilde} we make the replacement
 \be\label{replace}
 \tilde c_{\tilde\P}\nu_{\tilde\P}\quad \longrightarrow\quad
  \frac{|\tilde\C|}{|\tilde\P|}\tilde c_{\tilde\P}\nu_{\tilde\C}
 +\frac{|\P'|}{|\tilde\P|}\tilde c_{\tilde\P}\nu_{\P'},
 \ee
 then $\pi_k\tilde\mu$ is unchanged and a simple computation shows that
$S(\tilde\mu)$ does not increase.  By choice of $\tilde\mu$, then,
$\tilde\C$ must appear in \eqref{mutilde} with positive coefficient.  But
then the same computation shows that $S(\tilde\mu)$ must strictly decrease
under \eqref{replace}, contradicting its minimality.

To complete the proof we suppose that $N$ is the minimum number of cycles
needed to express an entropy minimizing extension of $\mu_k$ in the form
$\tilde\mu=\sum_{i=0}^Na_i\nu_{\C_i}$, $a_i>0$, assume that $N>2^k+1$, and
derive a contradiction.  Now certainly
$\mu_k=\sum_{i=0}^Na_i\pi_k\nu_{\C_i}$, and it then follows from the proof
in \cite{G} of the Carath\'eodory result mentioned in \crem{cycles}(b) that
(possibly after renumbering the cycles) we may also write
$\mu_k=\sum_{i=0}^{N-1}a^*_i\pi_k\nu_{\C_i}$ for some nonnegative
coefficients $a^*_i$.  Set $a_N^*=0$ and consider the measure
 \be
\tilde\mu(t)=\sum_{i=0}^Na_i(t)\nu_{\C_i}, \qquad
   a_i(t)=(1-t)a_i+ta^*_i,
 \ee
 and let $[u,1]$, with $u<0$, be the maximal interval such that
$a_i(t)\ge0$ for $i=0,\ldots,N$ if $t\in[u,1]$.  Note then that
$\pi_k\tilde\mu(t)=\mu_k$ for $t\in[u,1]$ and that the entropy
$S(\tilde\mu(t))$, which is concave in $t$, achieves its minimum at $t=0$,
and hence is constant. Thus $S(\tilde\mu(1))=S(\tilde\mu)$, contradicting
the minimality of $N$.  \end{proof}

 \begin{remark} (a) The argument of the second step in the proof of
\cthm{minentth} may be used to strengthen the conclusion expressed in
\eqref{muhat} by showing that in fact every entropy minimizing extension
of $\mu_k$ must be a {\it finite} convex combination of BPC measures.

 \smallskip\noindent
 (b) From \eqref{entnu} and \eqref{muhatf} it follows easily that the
entropy of a minimal entropy  extension of $\mu_k$ can be at most
$(2^k+1)[\log2^k+\log(2^k+1)]$.  \end{remark}

\cthm{minentth} reduces the search for a minimal entropy extension of a
given $\mu_k$ to a finite problem: for each set of $2^k+1$ cycles in
$G_k$ one tests whether $\mu_k$ is a convex combination of the measures
$\pi_k\nu_\C$ for cycles $\C$ in that set; if it is, the coefficients are
uniquely determined and one may compute the entropy of the corresponding
combination of the $\nu_\C$.  One then  chooses the minimum
entropy among the latter.  However, this is not practical for $k$ very
large, and we have no better approach at the moment.

There is, however, one case in which the minimal entropy extension of
$\mu_k$ may be obtained explicitly.  Suppose that $\mu_k$ is such that
there is a bijective measure preserving mapping $\phi:X_k\to X_k$ which for
each vertex $\xi$ of $G_k$, $\xi\in X_{k-1}$, maps the set of edges
entering $\xi$ to the set of edges leaving $\xi$, that is,
$\phi(\{0\xi,1\xi\})=\{\xi0,\xi1\}$.  Then $\phi$ determines a set of
pairwise edge-disjoint paths $\P_1,\ldots,\P_n$ by the condition that if
$\eta^{(1)},\ldots,\eta^{(|\P_i|)}$ are the edges of $\P_i$ in order then
$\eta^{(j+1)}=\phi(\eta^{(j)})$, $j=1,\ldots,|\P_i|-1$, and
$\phi(\eta^{(|\P_i|)})=\eta^{(1)}$.  Now because $\phi$ is measure
preserving, $\mu_k$ gives equal weight to the edges of each $\P_i$, and thus
the measure $\sum_{i=1}^n w_i|\P_i|\nu_{\P_i}$, where $w_i=\mu_k(\eta)$ for
any edge $\eta$ of $\P_i$, is an extension of $\mu_k$ with the same entropy
as $\mu_k$ and is thus the minimal entropy extension.  This is precisely
the case in which $\mu_k$ has a TI extension that is $(k+1)$-deterministic,
in the sense that is is supported on configurations that are determined by
any segment of length $k+1$.

A special case of the above occurs when $\mu_k$ is supported on a set of
pairwise vertex-disjoint cycles $\C_1,\ldots,\C_n$; this means that for
each vertex of $G_k$ at most one incoming and one outgoing edge is assigned
positive probability by $\mu_k$.  Then $\mu_k$ has the {\it unique} TI
extension $\sum_{i=1}^n w_i|\C_i|\nu_{\C_i}$, with $w_i$ as above.  In
particular, the maximal entropy extension of \csect{maxent} and the minimal
entropy extension of \cthm{minentth} are the same.  The total entropy of
the extension is the entropy of $\mu_k$, so that one sees from \eqref{ent}
(or directly) that this is also the entropy of the restriction $\mu_{k-1}$
of $\mu_k$ to $X_{k-1}$. The TI extension in this case is $k$-deterministic.

 \begin{remark} One might try to obtain a minimum entropy extension of
$\mu_k$ by a step-by-step procedure, extending first to a measure
$\mu_{k+1}$ on $X_{k+1}$, then to $\mu_{k+2}$ on $X_{k+2}$, etc.,
choosing at each step the minimum entropy extension for that step.  In
fact, there is a simple algorithm for each single step, but the procedure
may not produce the extension with the minimal total entropy, even for
the case $k=2$.  We omit details.
 \end{remark}

 \section{Nonextendible LTI measures\label{noextend}}

In this section we give examples to show that the local invariance
condition (LTI or, for the Bethe lattice, LI) is in general not sufficient
for extension to an invariant measure, unless $\Lambda$ is an interval in
$\bbz$, a strip in $\bbz^2$, a higher dimensional analogue of the latter,
or a full subgraph of the Bethe lattice.

\begin{example}\label{1d} (a) Let $\Lambda=\{0,1,3\}\subset\bbz$ and
$\Omega=\{\circ,\bullet\}$, and suppose that $\mu_\Lambda$ assigns
probability $1/2$ to each of the two configurations
$(\bullet,\bullet,\,\cdot\;,\circ)$ and $(\circ,\circ,\,\cdot\;,\bullet)$.
Then $\mu_\Lambda$ is LTI but does not even have a LTI extension to the
interval $\{0,1,2,3\}$, and hence not a TI extension to $\bbz$.

 \par\smallskip\noindent
 (b) We may give a similar example for the Bethe lattice, again with
$\Omega=\{\circ,\bullet\}$. In the notation of \csect{trees} we take $q=2$
and let $H$ be one of the connected subgraphs of $T_2$ with three vertices.
Then the measure which assigns probability $1/4$ to each of the
configurations \mini\circ\bullet\bullet, \mini\circ\circ\bullet,
\mini\bullet\bullet\circ, and \mini\bullet\circ\circ\ is LI but does not
extend to an LI measure on $X_2$.  \end{example}

Note that when $\Lambda\subset\bbz$ and $\Lambda_k$ is the minimal
interval containing $\Lambda$, the extendibility of an LTI measure from
$X_\Lambda$ to $X_k$, which is necessary and, by \cthm{jll}, sufficient
for extendibility to $X$, may be determined by solving a standard linear
programming feasibility problem.  If $H$ is a finite subgraph of $\bbt$
then a similar remark applies to the extension of an LI measure from
$X_H$ to $X_T$, where $T$ is the minimal full subgraph of $\bbt$
containing $H$.  Thus the problem of extension in one dimension, or on
the Bethe lattice, is decidable; this is in contrast to the situation in
higher dimension, as discussed in \csect{decide}.

\begin{example}\label{22} Let $\Lambda$ be a $2\times2$ square in
$\bbz^2$, let $\Omega=\{\circ,\bullet\}$, and let $\mu$ be the measure
which assigns probability $1/4$ to each of the configurations
$\eta^{(i)}$, $i=1,\ldots,4$, shown in \cfig{2by2}. $\mu$ is LTI, since
the marginal on the top, bottom, left, or right two sites of $\Lambda$ is
Bernoulli with parameter $1/2$, but $\mu$ cannot be extended to a measure
$\mu'$ on configurations on a $3\times3$ square $\Lambda'$ containing
$\Lambda$.  For suppose that $\mu'$ is such an extension; $\mu'$ can give
positive probability to a configuration $\eta'$ only if the restriction
of $\eta'$ to each of the four $2\times2$ subsquares of $\Lambda'$ is one
of the $\eta^{(i)}$.  Moreover, there must be an $\eta'$ with
$\mu'(\eta')>0$ having restriction $\eta^{(1)}$ to the lower left
subsquare, and $\eta'$ must then have restrictions $\eta^{(3)}$ and
$\eta^{(4)}$ to the upper left and lower right subsquares, respectively.
But then the restriction of $\eta'$ to upper right subsquare will not be
one of the $\eta^{(i)}$, whether or not the upper right hand site of
$\Lambda'$ is occupied in $\eta'$.  \end{example}

\begin{figure}[ht]
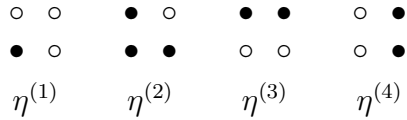

\centerline{\vbox{\beginpicture 
\setcoordinatesystem units <0.2truein,0.2truein> 
\setplotarea x from -0.5 to 10.5, y from -2 to 1.2
 \setlinear
 \put {$\bullet$} at 0 0
 \put {$\circ$} at 1 0
 \put {$\circ$} at 0 1
 \put {$\circ$} at 1 1
 \put {$\bullet$} at 3 0
 \put {$\bullet$} at 3 1
 \put {$\bullet$} at 4 0
 \put {$\circ$} at 4 1
 \put {$\circ$} at 6 0
 \put {$\circ$} at 7 0
 \put {$\bullet$} at 6 1
 \put {$\bullet$} at 7 1
 \put {$\circ$} at 9 0
 \put {$\circ$} at 9 1
 \put {$\bullet$} at 10 0
 \put {$\bullet$} at 10 1
 \put {$\eta^{(1)}$} [t] at 0.5 -0.8
 \put {$\eta^{(2)}$} [t] at 3.5 -0.8
 \put {$\eta^{(3)}$} [t] at 6.5 -0.8
 \put {$\eta^{(4)}$} [t] at 9.5 -0.8
\endpicture}}
\caption{\label{2by2}Configurations given nonzero probability in \cex{22}.  
Filled circles represent  occupied sites.}
\end{figure}

The measure $\mu$ in \cex{22} is not invariant under any of the symmetry
operations of the square; one could hope that the imposition of some symmetry
might give extensibility.  The next example shows that this is not the
case.

\begin{example}\label{33}Let $\Lambda$ be a $3\times3$ square in $\bbz^2$
and let $\mu$ be the measure which assigns probability $1/7$ to each of
the configurations $\eta^{(i)}$, $i=1,\ldots,7$, shown in \cfig{3by3}.
$\mu$ is easily seen to be LTI and is clearly invariant under any
symmetry of the square.  However, $\mu$ cannot be extended to any
$4\times4$ square $\Lambda'$ containing $\Lambda$.  The proof is similar
to the proof of nonextendibility given in \cex{22}: one checks, for
example, that there is no configuration $\eta'$ in $\Lambda'$ which
reduces to $\eta^{(2)}$ in the lower left hand $3\times3$ subsquare and
to one of the $\eta^{(i)}$ in all the other $3\times3$ subsquares.
\end{example}

\begin{figure}[ht]
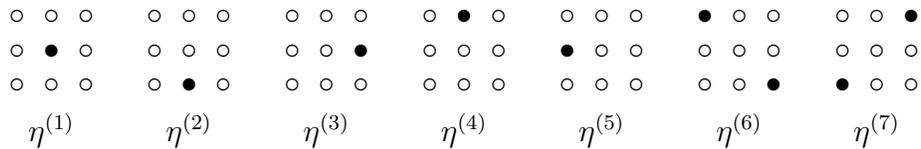

\centerline{\vbox{\beginpicture 
\setcoordinatesystem units <0.18truein,0.18truein> 
\setplotarea x from -0.5 to 10.5, y from -2 to 1.2
 \setlinear
 \put {$\circ$} at 0 0
 \put {$\circ$} at 0 1
 \put {$\circ$} at 0 2
 \put {$\circ$} at 1 0
 \put {$\bullet$} at 1 1
 \put {$\circ$} at 1 2
 \put {$\circ$} at 2 0
 \put {$\circ$} at 2 1
 \put {$\circ$} at 2 2
 \put {$\circ$} at 4 0
 \put {$\circ$} at 4 1
 \put {$\circ$} at 4 2
 \put {$\bullet$} at 5 0
 \put {$\circ$} at 5 1
 \put {$\circ$} at 5 2
 \put {$\circ$} at 6 0
 \put {$\circ$} at 6 1
 \put {$\circ$} at 6 2
 \put {$\circ$} at 8 0
 \put {$\circ$} at 8 1
 \put {$\circ$} at 8 2
 \put {$\circ$} at 9 0
 \put {$\circ$} at 9 1
 \put {$\circ$} at 9 2
 \put {$\circ$} at 10 0
 \put {$\bullet$} at 10 1
 \put {$\circ$} at 10 2
 \put {$\circ$} at 12 0
 \put {$\circ$} at 12 1
 \put {$\circ$} at 12 2
 \put {$\circ$} at 13 0
 \put {$\circ$} at 13 1
 \put {$\bullet$} at 13 2
 \put {$\circ$} at 14 0
 \put {$\circ$} at 14 1
 \put {$\circ$} at 14 2
 \put {$\circ$} at 16 0
 \put {$\bullet$} at 16 1
 \put {$\circ$} at 16 2
 \put {$\circ$} at 17 0
 \put {$\circ$} at 17 1
 \put {$\circ$} at 17 2
 \put {$\circ$} at 18 0
 \put {$\circ$} at 18 1
 \put {$\circ$} at 18 2
 \put {$\circ$} at 20 0
 \put {$\circ$} at 20 1
 \put {$\bullet$} at 20 2
 \put {$\circ$} at 21 0
 \put {$\circ$} at 21 1
 \put {$\circ$} at 21 2
 \put {$\bullet$} at 22 0
 \put {$\circ$} at 22 1
 \put {$\circ$} at 22 2
 \put {$\bullet$} at 24 0
 \put {$\circ$} at 24 1
 \put {$\circ$} at 24 2
 \put {$\circ$} at 25 0
 \put {$\circ$} at 25 1
 \put {$\circ$} at 25 2
 \put {$\circ$} at 26 0
 \put {$\circ$} at 26 1
 \put {$\bullet$} at 26 2
 \put {$\eta^{(1)}$} [t] at 1 -0.8
 \put {$\eta^{(2)}$} [t] at 5 -0.8
 \put {$\eta^{(3)}$} [t] at 9 -0.8
 \put {$\eta^{(4)}$} [t] at 13 -0.8
 \put {$\eta^{(5)}$} [t] at 17 -0.8
 \put {$\eta^{(6)}$} [t] at 21 -0.8
 \put {$\eta^{(7)}$} [t] at 25 -0.8
\endpicture}}
\caption{\label{3by3}Configurations given nonzero probability in \cex{33}.}
\end{figure}

The measures in the above examples appear rather special, in that they
assign zero probability to most configurations.  Recall from \csect{intro},
however, that the set of nonextendible measures (for fixed $\Lambda$) is
open; thus small perturbations of these measures will also be nonextendible.
Such perturbations may be taken to assign nonzero probability to every
configuration in $\Lambda$, to be be again LTI, and, for \cex{33}, to
have all the symmetries of the square.

\section{Undecidability\label{decide}}

As discussed above, the situation in $\bbz^d$ for $d\ge2$ is quite different
from that for $d=1$ or for the Bethe lattice.  In particular, there is no
formula similar to \eqref{MC} or \eqref{ext2} for an extension, and the LTI
condition is not sufficient for extendibility.  More surprisingly, the
extendibility problem for $\bbz^d$ with $d\ge2$ is in several senses
undecidable.  To discuss this we introduce the idea of {\it tiling}.  It
suffices to fix $\Lambda$ to be a $2\times2$ box in $\bbz^2$, and we will
do so throughout this section.  Let $Y$ be a finite subset of
$X_\Lambda=\Omega^\Lambda$ constructed with some $\Omega$; we call the
elements of $Y$ {\it tiles} and say that $Y$ {\it tiles the plane\/} if
there is a configuration in $X$ whose restriction to every $2\times2$ box
is an element of $Y$.  Here we do not consider $\Omega$ as fixed, so there
is an infinite number of possible sets $Y$ of tiles; given $Y$ we will
consider only tilings by $Y$, so that for all practical purposes $Y$
determines $\Omega$, and we will omit mention of the latter in what
follows.  Berger \cite{Berger} proved:

 \par\medskip
\begin{theorem}\label{berger}The tiling problem is undecidable: there is
no decision procedure which, given a set of tiles $Y$, determines whether
or not $Y$ tiles the plane.  \end{theorem}

 The statement of the theorem means that there is no algorithm (or Turing
machine) to which we can submit a tile set $Y$ and which will then yield
the result ``tiles'' when $Y$ in fact tiles and ``does not tile" when it
does not. Note that it is crucial here that the set of possible tile sets
$Y$ be infinite, since otherwise there would exist such a decision
procedure, though we might not know what it is: if there were $N$ tile sets
$Y_1,\ldots,Y_N$ then it would be one of the $2^N$ procedures which answer
the $N$ questions ``Does $Y_k$ tile?'' in all $2^N$ possible ways.

We should remark that Berger worked with a different sort of tile,
originally introduced by Wang \cite{Wang}, but it is easy to see
\cite{Schlijper4} that his result implies the above theorem.  A simplified
version of the proof of Berger's theorem is due to R. M. Robinson
\cite{Robinson}.

There is an alternative formulation of undecidability which follows from
Berger's result.

\begin{corollary}\label{SG} There exists a set of tiles $Y_*$ such that
$Y_*$ tiles the plane but cannot be proved to do so.  \end{corollary}

\begin{proof} It is known \cite{Robinson} that if $Y$ does not tile the
plane then there is some finite rectangle which it does not tile, so that
the fact that $Y$ does not tile is provable.  Thus if the conclusion of the
theorem were false an algorithm generating all provable theorems would
provide a decision procedure for the tiling problem.  \end{proof}

It should be noted that \ccor{SG} is not really stated precisely; whether
or not $Y$ can be proved to tile depends on the axiom system on which the
proof is to be based.  In fact, however, a tile set $Y_*$ satisfying the
theorem will always exist, but will in general depend on the (consistent
and sound) axiom system chosen.

 Schlijper \cite{Schlijper4} observed a connection between the tiling
 problem and the extendibility of measures:

 \begin{theorem} [Schlijper]\label{Schthm} A set of tiles $Y$ tiles the
plane if and only if there exists an extendible measure $\nu_Y$ on
$X_\Lambda$ with support $Y$.  \end{theorem}

Then an undecidability result follows:

\begin{corollary}\label{notile} There is no decision procedure
which, given a set of tiles $Y$, determines whether or not there exists
an extendible measure $\nu_\Lambda$ on $X_\Lambda$ with support $Y$.
Equivalently, there is no decision procedure which, given an LTI
measure $\mu_\Lambda$ on $X_\Lambda$, determines whether or not there
exists an extendible measure $\nu_\Lambda$ on $X_\Lambda$ with the same
support.  \end{corollary}

This result certainly suggests that there should be no decision procedure
for our original problem of determining whether or not a given LTI
measure $\mu_\Lambda$ is extendible, that is, that this problem is
undecidable.  However, since most of the measures $\mu_\Lambda$ are
presumably not expressible---not computable---most cannot in fact be
``given'' to the decision procedure to begin with.

Nevertheless, \ccor{notile} does lead to a conclusion which furnishes
a strong sense of undecidability for the original problem.  Suppose there
were a condition $C$ for extendibility which was ``useful'', that is,
simpler and more transparent than extendibility itself---perhaps something
in the nature of the LTI condition (although of course we know that this is
not a such a condition for $d\ge2$).  No matter how simple $C$ were,
however, there could exist no decision procedure which, given a set $Y$
of tiles, would determine whether or not there exists a measure
$\nu_\Lambda$ on $X_\Lambda$ with support $Y$ which satisfies $C$.

Another formulation of undecidability for the extension problem,
alternative to \ccor{notile} in the sense that \ccor{SG} is alternative to
\cthm{berger}, follows immediately from \ccor{SG} and \cthm{Schthm}:

\begin{corollary}\label{lastcor} There exists a set of tiles $Y_*$ such
that (1) there is an extendible measure $\nu_\Lambda$ on $X_\Lambda$ with
support $Y_*$, but (2) this fact cannot be proven.  \end{corollary}

Here, as for \ccor{SG}, $Y_*$ will in general depend on the axiom system
to be used in the proof.  It of course follows from \ccor{lastcor} that
if we specify a measure on $X_\Lambda$ with support $Y_*$ to arbitrary
but not perfect precision---that is, specify that it lie in some
arbitrarily small set of measures---then for some such specification
there exists an extendible measure satisfying it whose existence can't be
proven.  This suggests that there exists an extendible measure
$\nu_\Lambda$ on $X_\Lambda$ for which the proposition that it is
extendible can't be proven.  However this is not so clear, since the
relevant measure may not be suitably expressible and thus the desired
proposition may not exist.

As with \ccor{notile}, however, \ccor{lastcor} does have a consequence
which strongly suggests undecidability for the original problem: if there
were a condition $C$ for extendibility, ``useful'' in the sense discussed
above, then the statement that there exists a measure with support $Y_*$
which satisfies $C$ would be true but unprovable.

 \bigskip\noindent
 {\bf Acknowledgments.} The work J.L.L. was supported in part by NSF Grant
DMR 1104500 and AFOSR Grant FA9550-16-1-0037.  We thank A. C. D. van Enter
and M. Hochman for bringing to our attention previous work on this problem,
M. Hochman, M. Saks, S. Thomas, and A. C. D. van Enter for helpful
discussions, and D. Avis for making the computer program {\it lrs}
available to the public and for helpful advice on its use.

\end{document}

%% file: ptifig_trees2.tex
\newdimen\edge \edge0.27truein
\hbox to\hsize{\qquad
\beginpicture 
\setcoordinatesystem units <\edge,\edge> 
\setplotarea x from -0.5 to 0.5, y from -2.6 to 1
 \put {$\bullet$} at 0.0  0.0
 \put {$ k=0$} at 0.0  -2.5
\endpicture
\hfill
\beginpicture 
\setcoordinatesystem units <\edge,\edge> 
\setplotarea x from -1.1 to 1.1, y from -2.6 to 1
 \put {$\bullet$} at  0.5  0.0
 \put {$\bullet$} at -0.5  0.0
 \plot  -0.5  0.0 0.5  0.0 /
 \put {$ k=1$} at 0.0  -2.5
\endpicture
\hfill
\beginpicture 
\setcoordinatesystem units <\edge,\edge> 
\setplotarea x from -1.1 to 1.6, y from -2.6 to 1
 \put {$\bullet$} at  0.5  0.0
 \put {$\bullet$} at -0.5  0.0
 \put {$\bullet$} at 1.0  0.866
 \put {$\bullet$} at 1.0 -0.866
 \plot  -0.5  0.0 0.5  0.0 /
 \plot 1.0  0.866 0.5  0.0 /
 \plot 1.0  -0.866 0.5  0.0 /
 \put {$ k=2$} at 0.5  -2.5
\endpicture
\hfill
\beginpicture 
\setcoordinatesystem units <\edge,\edge> 
\setplotarea x from -1.5 to 2.0, y from -2.6 to 1.8
 \put {$\bullet$} at  0.5  0.0
 \put {$\bullet$} at -0.5  0.0
 \put {$\bullet$} at 1.0  0.866
 \put {$\bullet$} at 1.0 -0.866
 \put {$\bullet$} at -1.0  0.866
 \put {$\bullet$} at -1.0 -0.866
 \plot  -0.5  0.0 0.5  0.0 /
 \plot 1.0  0.866 0.5  0.0 /
 \plot 1.0  -0.866 0.5  0.0 /
 \plot -1.0  0.866 -0.5  0.0 /
 \plot -1.0  -0.866 -0.5  0.0 /
 \put {$ k=3$} at 0.0  -2.5
\endpicture
\hfill
\beginpicture 
\setcoordinatesystem units <\edge,\edge> 
\setplotarea x from -1.1 to 2.1, y from -2.6 to 1.8
 \put {$\bullet$} at  0.5  0.0
 \put {$\bullet$} at -0.5  0.0
 \put {$\bullet$} at 1.0  0.866
 \put {$\bullet$} at 1.0 -0.866
 \put {$\bullet$} at -1.0  0.866
 \put {$\bullet$} at -1.0 -0.866
 \put {$\bullet$} at  2.0  0.866
 \put {$\bullet$} at  0.5  1.732
 \put {$\bullet$} at  2.0  -0.866
 \put {$\bullet$} at  0.5  -1.732
 \plot  -0.5  0.0 0.5  0.0 /
 \plot 1.0  0.866 0.5  0.0 /
 \plot 1.0  -0.866 0.5  0.0 /
 \plot -1.0  0.866 -0.5  0.0 /
 \plot -1.0  -0.866 -0.5  0.0 /
 \plot 1.0   0.866 2.0   0.866 /
 \plot 1.0   0.866 0.5  1.732 /
 \plot 1.0  -0.866 2.0  -0.866 /
 \plot 1.0  -0.866 0.5  -1.732 /
 \put {$ k=4$} at 0.5  -2.5
\endpicture
\hfill}